\documentclass[11pt]{amsart}
\usepackage{amsmath}
\usepackage{amssymb}
\usepackage{harpoon}
\usepackage[dvips]{epsfig}
\theoremstyle{plain}
\newtheorem{theorem}{Theorem}[section]

\theoremstyle{definition}

\numberwithin{equation}{section}

\begin{document}

\title[Existence of Black Holes Due to Concentration of Angular Momentum] {Existence of Black Holes Due to Concentration of Angular Momentum}

\author[Khuri]{Marcus A. Khuri}
\address{Department of Mathematics\\
Stony Brook University\\
Stony Brook, NY 11794, USA}
\email{khuri@math.sunysb.edu}

\thanks{The author acknowledges the support of
NSF Grant DMS-1308753.}

\begin{abstract}
We present a general sufficient condition for the formation of black holes due to concentration of angular momentum. This is expressed in the form of a universal inequality, relating the size and angular momentum of bodies, and is proven in the context of axisymmetric initial data sets for the Einstein equations which satisfy an appropriate energy condition. A brief comparison is also made with more traditional black hole existence criteria based on concentration of mass.
\end{abstract}
\maketitle

\section{Introduction}
\label{sec1} \setcounter{equation}{0}
\setcounter{section}{1}

The Trapped Surface Conjecture \cite{Seifert} and Hoop Conjecture \cite{Thorne} are concerned with
the folklore belief that if enough matter and/or gravitational energy are present
in a sufficiently small region, then the system must collapse to a black hole.
In more concrete terms, this belief is often realized by establishing a statement
of the following form.  Let $\Omega$ be a compact spacelike hypersurface satisfying an
appropriate energy condition in a spacetime $\mathcal{M}$.  There
exists a universal constant $\mathcal{C}>0$ such that if
\begin{equation}\label{1}
\mathrm{Mass}(\Omega)>\mathcal{C}\cdot\mathrm{Size}(\Omega),
\end{equation}
then $\Omega$ is either enclosed by, or has a nontrivial intersection with, a closed trapped surface.  The point is that the presence of a closed trapped surface implies that the spacetime $\mathcal{M}$ contains a singularity (or more precisely is null geodesically incomplete) by the Hawking-Penrose
Singularity Theorems \cite{HawkingEllis}, and assuming Cosmic Censorship \cite{Penrose} must therefore contain a black hole. Although this problem is well-studied, the general case is still open. In particular, previous results \cite{BeigOMurchadha,BizonMalecOMurchadha1,BizonMalecOMurchadha2,
Flanagan,Khuri,Malec1,Malec2,Wald1} require special auxiliary conditions, for instance assuming that the spacelike slice is spherically symmetric or maximal, whereas others \cite{Eardley,SchoenYau2,Yau} are not meaningful for slices with small extrinsic curvature.

While it is natural, based on intuition, to suggest that high concentrations of matter and/or gravitational energy should lead to black hole formation, here we will propose a 
less intuitive criterion for gravitational collapse. Namely, we will show that high concentrations of angular momentum alone are enough to induce collapse. In analogy with \eqref{1}, this will be realized by proving the
existence of a universal constant $\mathcal{C}>0$ such that if
\begin{equation}\label{2}
|\mathcal{J}(\Omega)|>\mathcal{C}\cdot\mathcal{R}^{2}(\Omega),
\end{equation}
then $\Omega$ is either enclosed by, or has a nontrivial intersection with, a closed trapped surface.
Here $\mathcal{J}(\Omega)$ represents total angular momentum, and $\mathcal{R}(\Omega)$ is a certain
radius which measures the size of $\Omega$. Thus, if a rotating body possesses enough angular momentum
and is sufficiently small, then the system must collapse to form a black hole. To the best of the author's knowledge, this is a new and previously uninvestigated process by which black holes may
form. In what follows, we will give a heuristic justification of \eqref{2} as well as a rigorous proof for axially symmetric rotating systems.

\section{Heuristic Evidence and Precise Formulation}
\label{sec2} \setcounter{equation}{0}
\setcounter{section}{2}

In \cite{Dain}, Dain introduced an inequality relating the size and angular momentum for General Relativistic bodies of the form
\begin{equation}\label{3}
\mathcal{R}^{2}(\Omega)\gtrsim\frac{G}{c^{3}}|\mathcal{J}(\Omega)|,
\end{equation}
where $G$ is the gravitational constant, $c$ is the speed of light, and $\gtrsim$ represents an
order of magnitude; the precise constant depends on the choice of radius $\mathcal{R}$. He presented
heuristic arguments which imply \eqref{3}, and which are based on the following four assumptions
and principles:
\medskip

(i) the body $\Omega$ is not contained in a black hole,\medskip

(ii) the speed of light $c$ is the maximum speed,\medskip

(iii) the reverse inequality of \eqref{1} holds for bodies $\Omega$ which
are not contained in a black hole,\medskip

(iv) inequality \eqref{3} holds for black holes.\medskip

Now suppose that \eqref{2} holds. Then \eqref{3} is violated, and hence one of the above four
statements cannot be valid. Since (ii) is firmly established, (iii) is closely tied to the Trapped Surface/Hoop Conjecture and is expected to hold, and (iv) has been proven \cite{Dain0},
it follows that assumption (i) should be false. In general terms, we conclude that if a body satisfies \eqref{2}, this should indicate the presence of a black hole.

A precise version of this conclusion will now be described, and then proven below.
Consider an initial data set $(M, g, k)$ for the Einstein equations. This consists of a 3-manifold $M$, (complete) Riemannian metric $g$, and symmetric 2-tensor $k$ representing the
extrinsic curvature (second fundamental form) of the embedding into spacetime, which satisfy the constraint equations
\begin{align}\label{4}
\begin{split}
\frac{16\pi G}{c^{4}}\mu &= R+(Tr_{g} k)^{2}-|k|^{2},\\
\frac{8\pi G}{c^{4}} J^{i} &= \nabla_{j}(k^{ij}-(Tr_{g} k)g^{ij}).
\end{split}
\end{align}
Here $\mu$ and $J$ are the energy and momentum densities of the matter fields, respectively, and $R$ is the scalar curvature of $g$. In terms of the 4-dimensional stress-energy tensor $T_{ab}$ we have $\mu=T_{ab}n^{a}n^{b}$ and $J_{i}=T_{ia}n^{a}$, where $n^{a}$ denotes the timelike unit normal to the slice $M$. A body $\Omega$ is a connected open subset of $M$ with compact closure and smooth boundary $\partial\Omega$.

We say that the initial data are axially symmetric if the group of isometries of the Riemannian manifold $(M,g)$ has a subgroup isomorphic to $U(1)$, and that the remaining quantities defining the initial data are invariant under the $U(1)$ action. In particular, if $\eta^{i}$ denotes the Killing field associated with this symmetry, then
\begin{equation}\label{5}
\mathfrak{L}_{\eta}g=\mathfrak{L}_{\eta}k=\mathfrak{L}_{\eta}\mu=\mathfrak{L}_{\eta}J=0,
\end{equation}
where $\mathfrak{L}_{\eta}$ denotes Lie differentiation. Furthermore, the total angular momentum  of the body $\Omega$ is defined by
\begin{equation}\label{6}
\mathcal{J}(\Omega)=\frac{1}{c}\int_{\Omega}J_{i}\eta^{i} d\omega_{g}.
\end{equation}
Axisymmetry is imposed primarily to obtain a suitable and well-defined notion of angular momentum for bodies. Note that in this setting gravitational waves have no angular momentum, so all the angular momentum is contained in the matter sources. Without this assumption quasi-local angular momentum is difficult to define \cite{Szabados}.
It will also be assumed that the following version of the dominant energy condition
holds on $\Omega$, namely
\begin{equation}\label{7}
\mu\geq|\vec{J}|+|J(e_{3})|
\end{equation}
where $(e_{1},e_{2},e_{3}=|\eta|^{-1}\eta)$ is an orthonormal frame field on $M$ and $\vec{J}=J(e_{1})e_{1}+J(e_{2})e_{2}$. Note that this is a stronger version of
the classical dominant energy condition which states
\begin{equation}\label{8}
\mu\geq|J|=\sqrt{|\vec{J}|^{2}+J(e_{3})^{2}}.
\end{equation}

Let us now consider how to measure the size of the body $\Omega$. A particularly pertinent measure
in the current setting, is a homotopy radius defined by Schoen and Yau in \cite{SchoenYau2}, which played a crucial role in their criterion for the existence of black holes due to concentration of matter. The Schoen/Yau radius, $\mathcal{R}_{SY}(\Omega)$, may be described as the radius of the largest torus that can be embedded in $\Omega$. More specifically, let $\Gamma$ be a simple closed curve which bounds a disk in $\Omega$, and let $r$ denote the largest distance from $\Gamma$ such that the set of all points within this distance forms a torus embedded in $\Omega$. Then $\mathcal{R}_{SY}(\Omega)$ is defined to be the largest distance $r$ among all curves $\Gamma$. For example, if $B_{\rho}$ is a ball of radius $\rho$ in flat space, then $\mathcal{R}_{SY}(B_{\rho})=\rho/2$. In analogy with \cite{Dain}, we define the radius that appears in \eqref{2} by
\begin{equation}\label{9}
\mathcal{R}(\Omega)=\frac{\left(\int_{\Omega}|\eta|d\omega_{g}\right)^{1/2}}{\mathcal{R}_{SY}(\Omega)}.
\end{equation}

With these definitions of angular momentum $\mathcal{J}$ and radius $\mathcal{R}$, we obtain a precise formulation of inequality \eqref{2}, save for the universal constant $\mathcal{C}$ to be described below. It will be shown that this inequality implies the existence of a closed trapped surface, or more accurately an apparent horizon. Recall that the strength of the gravitational field in the vicinity of a
2-surface $S\subset M$ may be measured by the null expansions
\begin{equation}\label{10}
\theta_{\pm}:=H_{S}\pm Tr_{S}k,
\end{equation}
where $H_{S}$ is the mean curvature with respect to the unit
outward normal. The null expansions measure the rate of change of area for a shell of light
emitted by the surface in the outward future direction
($\theta_{+}$), and outward past direction ($\theta_{-}$).  Thus
the gravitational field is interpreted as being strong near
$S$ if $\theta_{+}< 0$ or $\theta_{-}< 0$, in which case
$S$ is referred to as a future (past) trapped surface.
Future (past) apparent horizons arise as boundaries of future
(past) trapped regions and satisfy the equation $\theta_{+}=0$
($\theta_{-}=0$). Apparent horizons may be thought of as quasi-local notions of event horizons,
and in fact, assuming Cosmic Censorship, they must generically be contained inside
black holes \cite{Wald}.

\section{Inequality Between Size and Angular Momentum for Bodies}
\label{sec3} \setcounter{equation}{0}
\setcounter{section}{3}

In order to establish the existence of apparent horizons when inequality \eqref{2} is satisfied, we will utilize a device employed by Schoen and Yau \cite{SchoenYau2}. Namely, they showed that if a certain differential equation does not possess a regular solution, then an apparent horizon must be present in the initial data. This so called Jang equation is given in local coordinates by
\begin{equation}\label{11}
\left(g^{ij}-\frac{f^{i}f^{j}}{1+|\nabla f|^{2}}\right)\left(\frac{\nabla_{ij}f}
{\sqrt{1+|\nabla f|^{2}}}-k_{ij}\right)=0,
\end{equation}
where $f^{i}=g^{ij}\nabla_{j}f$. Geometrically, this expression is equivalent to the apparent horizon equation, but in the 4-dimensional product manifold $\mathbb{R}\times M$. When regular solutions do not exist, the graph $t=f(x)$ blows-up and approximates a cylinder over an apparent horizon in the base manifold $M$ (see \cite{HanKhuri}, \cite{SchoenYau1}). Whether or not the solution blows-up, is related to concentration of scalar curvature or rather matter density for the induced metric, $\overline{g}_{ij}=g_{ij}+\nabla_{i}f\nabla_{j}f$, on the graph. In this regard, it is important to have an explicit
formula \cite{BrayKhuri1,BrayKhuri2,SchoenYau1} for the scalar curvature of this metric, namely
\begin{equation}\label{12}
\overline{R}=\frac{16\pi G}{c^{4}}(\mu-J(v))+
|h-k|_{\overline{g}}^{2}+2|q|_{\overline{g}}^{2}
-2div_{\overline{g}}(q),
\end{equation}
where $h$ is the second fundamental form of the graph, $div_{\overline{g}}$
is the divergence operator with respect to $\overline{g}$, and $q$ and
$v$ are 1-forms given by
\begin{equation}\label{13}
v_{i}=\frac{f_{i}}{\sqrt{1+|\nabla f|^{2}}},\text{
}\text{ }\text{ }\text{ }
q_{i}=\frac{f^{j}}{\sqrt{1+|\nabla f|^{2}}}(h_{ij}-k_{ij}).
\end{equation}

Suppose now that the Jang equation has a regular solution over $\Omega$. One way to measure the
concentration of scalar curvature on this region is to estimate
the first Dirichlet eigenvalue, $\lambda_{1}$, of the operator $\Delta_{\overline{g}}-\frac{1}{2}\overline{R}$; here $\Delta_{\overline{g}}=\overline{g}^{ij}\overline{\nabla}_{ij}$ is the Laplace-Beltrami operator.
Let $\phi$ be the corresponding first eigenfunction, then
\begin{equation}\label{14}
\lambda_{1}=\frac{\int_{\Omega}\left(|\overline{\nabla}\phi|^{2}
+\frac{1}{2}\overline{R}\phi^{2}\right)d\omega_{\overline{g}}}
{\int_{\Omega}\phi^{2}d\omega_{\overline{g}}}.
\end{equation}
Notice that in the expression for the scalar curvature \eqref{12}, only the divergence term yields a potentially negative contribution to the quotient.
However, after applying the divergence theorem, we find that this term is dominated by the two
nonnegative terms $|\overline{\nabla}\phi|^{2}$ and $|q|_{\overline{g}}^{2}\phi^{2}$. It follows that
\begin{equation}\label{15}
\lambda_{1}\geq\frac{8\pi G}{c^{4}}\frac{\int_{\Omega}(\mu-J(v))\phi^{2}d\omega_{\overline{g}}}
{\int_{\Omega}\phi^{2}d\omega_{\overline{g}}}.
\end{equation}
In light of the axial symmetry \eqref{5}, if the Jang solution $f$ possesses axially symmetric
boundary conditions, then $\mathfrak{L}_{\eta}f=0$ on $\Omega$. In particular, the 1-form $v$ has
no component in the $\eta$ direction, which implies that $\mu-J(v)\geq\mu-|\vec{J}|$. It follows that
\begin{equation}\label{16}
\lambda_{1}\geq\frac{8\pi G}{c^{4}}\frac{\int_{\Omega}(\mu-|\vec{J}|)\phi^{2}d\omega_{\overline{g}}}
{\int_{\Omega}\phi^{2}d\omega_{\overline{g}}}=:\Lambda.
\end{equation}

With a lower bound for the first eigenvalue in hand, we may apply Proposition 1 from \cite{SchoenYau2} to conclude
\begin{equation}\label{17}
\overline{\mathcal{R}}_{SY}(\Omega)\leq \sqrt{\frac{3}{2}}\frac{\pi}{\sqrt{\Lambda}},
\end{equation}
where $\overline{\mathcal{R}}_{SY}$ denotes the Schoen/Yau radius with respect to the metric $\overline{g}$. Observe that since $\overline{g}\geq g$, we have $\overline{\mathcal{R}}_{SY}
\geq\mathcal{R}_{SY}$. Moreover, multiplying and dividing $\Lambda^{-1}$ by the quantity
$\int_{\Omega}|\eta|d\omega_{g}\left(\int_{\Omega}(\mu-|\vec{J}|)|\eta|d\omega_{g}
\right)^{-1}$ yields
\begin{equation}\label{18}
\Lambda^{-1}\leq\frac{c^{4}\mathcal{C}_{0}}{8\pi G}\frac{\int_{\Omega}|\eta|d\omega_{g}}
{\int_{\Omega}(\mu-|\vec{J}|)|\eta|d\omega_{g}},
\end{equation}
where
\begin{equation}\label{19}
\mathcal{C}_{0}=
\frac{\max_{\Omega}\left(\mu-|\vec{J}|\right)}
{\min_{\Omega}\left(\mu-|\vec{J}|\right)}
\end{equation}
if $\mu-|\vec{J}|>0$ in $\Omega$, and $\mathcal{C}_{0}=\infty$ if $\mu-|\vec{J}|$ vanishes at some point of $\Omega$. Hence
\begin{equation}\label{20}
\int_{\Omega}(\mu-|\vec{J}|)|\eta|d\omega_{g}\leq
\frac{3\pi c^{4}\mathcal{C}_{0}}{16G}\frac{\int_{\Omega}|\eta|d\omega_{g}}
{\mathcal{R}_{SY}^{2}(\Omega)}.
\end{equation}
All together these arguments produce a general relation between the size and angular momentum of bodies.

\begin{theorem}\label{thm1}
Let $(M,g,k)$ be an axially symmetric initial data set
which contains no compact apparent horizons. Assume that either $M$ is asymptotically flat, or has a strongly untrapped boundary, that is $H_{\partial M}>|Tr_{\partial M}k|$. Then for any body $\Omega\subset M$
satisfying the energy condition \eqref{7}, the following
inequality holds
\begin{equation}\label{21}
|\mathcal{J}(\Omega)|\leq\frac{3\pi c^{3}\mathcal{C}_{0}}{16G}\mathcal{R}^{2}(\Omega).
\end{equation}
\end{theorem}

\begin{proof}
The conditions on the boundary of $M$ or its asymptotics guarantee the existence of a strongly untrapped 2-surface. For instance, if $M$ is asymptotically flat, then a large coordinate sphere in the asymptotic end will be strongly untrapped. This property allows one to solve the Dirichlet boundary
value problem \cite{SchoenYau2} for the Jang equation \eqref{11}, with $f=0$ on $\partial M$ or on an appropriate coordinate sphere in the asymptotic end. The solution $f$ will then be axisymmetric. Moreover, the absence of apparent horizons ensures that $f$ is a regular solution. We may then apply the arguments preceding this theorem to obtain \eqref{20}.

Now observe that with the help of the energy condition \eqref{7},
\begin{align}\label{22}
\begin{split}
|\mathcal{J}(\Omega)|\leq\frac{1}{c}\int_{\Omega}|J(\eta)|d\omega_{g}
&=\frac{1}{c}\int_{\Omega}|J(e_{3})||\eta|d\omega_{g}\\
&=\frac{1}{c}\int_{\Omega}\left[|J(e_{3})|+|\vec{J}|-\mu
+(\mu-|\vec{J}|)\right]|\eta|d\omega_{g}\\
&\leq\frac{1}{c}\int_{\Omega}(\mu-|\vec{J}|)|\eta|d\omega_{g}.
\end{split}
\end{align}
Combining \eqref{20} and \eqref{22} produces the desired result.
\end{proof}

This theorem is of independent interest, and generalizes the main result of \cite{Dain} in two ways. That is, Dain's inequality between the size and angular momentum of bodies required two strong hypotheses, namely that the initial data are maximal $Tr_{g}k=0$ and that the matter density $\mu$ is constant. Here we have removed both of these hypotheses at the expense of a slightly weaker inequality. More precisely, when $\mu-|\vec{J}|$ is constant, the two inequalities may be compared directly. The universal constant obtained by Dain, $\frac{\pi c^{3}}{6G}$, is less than the universal
constant of \eqref{21}, $\frac{3\pi c^{3}}{16G}$. Recently, other related inequalities have been derived by Reiris \cite{Reiris}, which also require the maximal hypothesis.

\section{Criterion for the Existence of Black Holes}
\label{sec4} \setcounter{equation}{0}
\setcounter{section}{4}

The proof of Theorem \ref{thm1} naturally leads to a black hole existence result, due to its reliance
on solutions of the Jang equation. As noted, this technique for producing black holes was originally
exploited by Schoen and Yau \cite{SchoenYau2}. Namely, if the reverse inequality of \eqref{21} holds, then we must conclude that the Jang equation does not admit a regular solution. This implies that the Jang solution blows-up and ensures the presence of an apparent horizon. We now state the main result.

\begin{theorem}\label{thm2}
Let $(M,g,k)$ be an axially symmetric initial data set, such that either $M$ is asymptotically flat, or has a strongly untrapped boundary, that is $H_{\partial M}>|Tr_{\partial M}k|$. If $\Omega\subset M$
is a bounded region satisfying the energy condition \eqref{7}, with
\begin{equation}\label{23}
|\mathcal{J}(\Omega)|>\frac{3\pi c^{3}\mathcal{C}_{0}}{16G}\mathcal{R}^{2}(\Omega),
\end{equation}
then $M$ contains an apparent horizon of spherical topology and in particular contains a closed
trapped surface.
\end{theorem}

It should be observed that Theorems \ref{thm1} and \ref{thm2} are independent of the particular
matter model, and only require an energy condition which prevents the matter from traveling faster than the speed of light.

Whether or not such a process, by which high concentrations of angular momentum leads to gravitational collapse, can occur in nature, appears to be an interesting open problem.
On the theoretical side, it is important to understand the types of geometries which admit
an inequality of the form \eqref{23}. In this regard we note that the inequality cannot hold
in the maximal case. This is due to the fact, explained at the
end of the previous section, that the constant in Dain's inequality \cite{Dain} is smaller than
the constant in \eqref{23}. This is analogous to the situation with Schoen and Yau's criterion
for black hole formation, in which a stronger inequality holds for bodies in the maximal case,
thus preventing such initial data from satisfying their hypotheses for the existence of trapped
surfaces. Thus, the geometries which satisfy the Schoen/Yau criterion require large amounts of
extrinsic curvature; see \cite{Malec2} for a discussion concerning this topic. We expect that the
same holds true for Theorem \ref{thm2}.

Let us now compare the above result with more traditional black hole existence criteria based on concentration of mass. In general terms, we have shown that a body undergoes gravitational collapse if its total angular momentum and radius satisfy an inequality of the form
\begin{equation}\label{24}
\frac{G}{c^{3}}|\mathcal{J}|\gtrsim\mathcal{R}^{2},
\end{equation}
whereas a version of the Hoop Conjecture asserts that
\begin{equation}\label{25}
\frac{2Gm}{c^{2}}>\mathcal{R}
\end{equation}
is sufficient for collapse, where $m$ denotes rest mass and the left-hand side is the Schwarzschild radius $\mathcal{R}_{s}$. Consider the fastest
spinning pulsar \cite{Pulsar} known to date, PSR J1748-2446ad. Its angular velocity is $\omega\approx4.5\times 10^{3}rad\text{ }\!s^{-1}$, it has a radius of $\mathcal{R}\approx 15km$, and consists of two solar masses $m\approx2M_{\odot}=4\times 10^{30}kg$. It follows that
\begin{equation}\label{26}
\sqrt{\frac{G}{c^{3}}|\mathcal{J}|}
=\sqrt{\frac{G}{c^{3}}m\omega\mathcal{R}^{2}}\approx 3.2km
\text{ }\text{ }\text{ }\text{ }\text{ }\text{ }\text{ }\text{ and }\text{ }\text{ }\text{ }\text{ }\text{ }\text{ }\text{ }
\frac{2Gm}{c^{2}}\approx5.9km.
\end{equation}
Upon comparing these values with the radius, we find that although it is close, neither \eqref{24} nor \eqref{25} is satisfied, as expected. Moreover, the similarity of the values in \eqref{26} seems to suggest that the criteria \eqref{24} and \eqref{25} may apply in similar regimes (at least for astronomical objects), however this is dependent on the magnitude of the optimal constant in \eqref{24} which is not addressed in this paper. It is thus a question worthy of further investigation to determine the optimal constant.

We may also compare the criteria \eqref{24} and \eqref{25} in the realm of elementary particles. It turns out that here, \eqref{24} offers additional insight due to the quantization of angular momentum. More precisely, for a particle of spin $s$ its angular momentum is given by
\begin{equation}\label{27}
|\mathcal{J}|=\sqrt{s(s+1)}\hbar,
\end{equation}
where $\hbar=1.05\times 10^{-30}cm^{2}s^{-1}kg$ is Planck's constant. According to \eqref{24}, in order for such a particle to remain stable gravitationally, its radius should be bounded below by a multiple of $[s(s+1)]^{1/4}l_{p}$ where $l_{p}=\sqrt{\frac{G\hbar}{c^{3}}}$
is the Planck length. A similar, although somewhat less convincing argument is known, from which one may derive the same conclusion based on criteria \eqref{25}. Namely, the gravitational stability of a particle implies that its radius should not be less than the Schwarzschild radius, $\mathcal{R}\geq\mathcal{R}_{s}$. Moreover, the length scale at and below which quantum effects become very important is given by the Compton wavelength
$\lambda=\frac{\hbar}{mc}$, and hence we should have $\mathcal{R}_{s}\geq\lambda$. This, however, implies that $\mathcal{R}_{s}\geq\sqrt{2}l_{p}$, which again yields a lower bound for the radius of the particle in terms of the Planck length.

\end{document}